\documentclass[a4paper]{article} 
\usepackage{amsmath,amssymb,amsfonts,amsthm} 
\usepackage{verbatim}
\usepackage{layout} 
\numberwithin{equation}{section}
\input{Macros.tex}

\newcommand{\comm}[1]{{#1}}

\newcommand\RR{{\mathbb R}}
\newcommand\TT{{\mathbb T}}
\newcommand\ZZ{{\mathbb Z}}
\newcommand\CC{{\mathbb C}}

\newcommand{\mc}[1]{{\mathcal #1}}
\newcommand{\mf}[1]{{\mathfrak #1}}

\newcommand{\bb}[1]{{\mathbb #1}}

\begin{document}
\title{Small perturbation of a disordered harmonic chain by a noise and an anharmonic potential}
\author{
C\'edric Bernardin\footnote{
Universit\'e de Lyon and CNRS, UMPA, UMR-CNRS 5669, ENS-Lyon,
46, all\'ee d'Italie, 69364 Lyon Cedex 07 - France.
E-mail: Cedric.Bernardin@umpa.ens-lyon.fr
} ,
Fran\c cois Huveneers\footnote{
CEREMADE, 
Universit\' e de Paris-Dauphine,
Place du Mar\'echal De Lattre De Tassigny,
75775 PARIS CEDEX 16 - FRANCE.
E-mail: huveneers@ceremade.dauphine.fr.
}
}
\maketitle 
\begin{abstract}
We study the thermal properties of a pinned disordered harmonic chain weakly perturbed by a noise and an anharmonic potential. 
The noise is controlled by a parameter $\lambda \rightarrow 0$, 
and the anharmonicity by a parameter $\lambda' \le \lambda$. 
Let $\kappa$ be the conductivity of the chain, defined through the Green-Kubo formula.
Under suitable hypotheses, we show that $\kappa = \mathcal O (\lambda)$
and, in the absence of anharmonic potential, that $\kappa \sim \lambda$. 
This is in sharp contrast with the ordered chain for which $\kappa \sim 1/\lambda$, 
and so shows the persistence of localization effects for a non-integrable dynamics.
\end{abstract}

\vspace*{9cm}

\pagebreak

\section{Introduction}
The mathematically rigorous derivation of macroscopic thermal properties of solids, 
starting from their microscopic description, 
is a serious challenge (\cite{bon}, \cite{lep}).
On the one hand, numerous experiments and numerical simulations show that, for a wide variety of materials, 
the heat flux is related to the gradient of temperature through a simple relation known as Fourier's law:
\begin{equation*}
J \; = \; - \kappa (T) \, \nabla T \, ,
\end{equation*}
where $\kappa (T)$ is the thermal conductivity of the solid. 
On the other hand, the mathematical understanding of this phenomenological law from the point of view of statistical mechanics is still lacking. 

A one-dimensional solid can be modelled by a chain of oscillators, 
each of them being possibly pinned by an external potential,
and interacting through a nearest neighbour coupling.
The case of homogeneous harmonic interactions can be readily analysed,
but it has been realized that this very idealized solid behaves like a perfect conductor, and so violates Fourier's law (\cite{rie}).
{\comm{To take into account the physical observations}}, it is thus needed to consider more elaborate models, where ballistic transport of energy is broken. 
Here are two possible directions. 

On the one hand, adding some anharmonic interactions can drastically affect the conductivity of the chain (\cite{aok}, \cite{luk}). 
Unfortunately, the rigorous study of anharmonic chains is in general out of reach, and even numerical simulations do not lead to completely unambiguous conclusions. 
In order to draw some clear picture, 
anharmonic interactions are mimicked in \cite{bas}\cite{ber2} by a stochastic noise that preserves total energy and possibly total momentum.
The thermal behaviour of anharmonic solids is, at a qualitative level, correctly reproduced by this partially stochastic model. 
By instance, the conductivity of the one-dimensional chain is shown to be positive and finite if the chain is pinned, and to diverge if momentum is conserved. 

On the other hand, {\comm{another element that can affect the conductivity of an harmonic chain is impurities}}. 
In \cite{rub} and \cite{cas}, an impure solid is modelled by a disordered harmonic chain, 
where the masses of the atoms are random. 
In these models, localization of eigenmodes induces a dramatic fall off of the conductivity. 
In the presence of everywhere onsite pinning, it is known that the chain behaves like a perfect insulator (see Remark 1 after Theorem \ref{theo: anharmonic}).
The case of unpinned chain is more delicate, and turns out to depend on the boundary conditions (\cite{dha3}).
The principal cases have been rigorously analysed in \cite{ver} and \cite{aja}. 

The thermal conductivity of an harmonic chain perturbed by both disorder and anharmonic interactions is a topic of both practical and mathematical interest. 
We will in the sequel only consider a one-dimensional disordered chain with everywhere on-site pinning. 
Doing so we avoid the pathological behaviour of unpinned one-dimensional chains, 
and we focus on a case where the distinction between ordered and disordered harmonic chain is the sharpest. 
We will consider the joint action of a noise and an anharmonic potential ; 
we call $\lambda$ the parameter controlling the noise, and $\lambda'$ the parameter controlling the anharmonicity (see Subsection \ref{subsec: model} below).

The disordered harmonic chain is an integrable system where localization of the eigenmodes can be studied rigorously (\cite{kun}). 
However, if some anharmonic potential is added, very few {\comm{is}} known about the persistence of localization effects. 
In \cite{dha2}, it is shown through numerical simulations that an even small amount of anharmonicity leads to a normal conductivity, 
destroying thus the localization of energy. 
In \cite{oga}, an analogous situation is studied and similar conclusions are reached. 
This is confirmed rigorously in \cite{ber}, if the anharmonic interactions are replaced by a stochastic noise preserving energy.
Nothing however is said there about the conductivity as $\lambda \rightarrow 0$. 
Later, this partially stochastic system has been studied in \cite{dha}, where numerical simulations indicate that $\kappa \sim \lambda$ as $\lambda \rightarrow 0$. 

Let us mention that, although the literature on the destruction of localized states seems relatively sparse in the context of thermal transport, 
much more is to find in that of Anderson's localization and disordered quantum systems
(see \cite{bask} and references in \cite{bask}\cite{dha}). 
There as well however, few analytical results seem to be available. 
Moreover, the interpretation of results from these fields to the thermal conductivity of solids is delicate, in part because many studies deal with systems at zero temperature: 
the time evolution of an initially localized wave packet. 

The main goal of this article is to establish that disorder strongly influences the thermal conductivity of a harmonic chain, 
when both a small noise and small anharmonic interactions are added. 
We will always assume that $\lambda' \le \lambda$, meaning that the noise is the dominant perturbative effect. 
Our main results, 
stated in Theorems \ref{theo: anharmonic} and \ref{theo: harmonic} below, 
are that $\kappa = \mathcal O (\lambda)$ as $\lambda \rightarrow 0$, and that $\kappa \sim \lambda$ if $\lambda' = 0$.
Strictly speaking, our results do not imply anything about the case where $\lambda' >0$ and $\lambda = 0$. 
However, in the regime we are dealing with, the noise is expected to produce interactions between localized modes, and so to increase the conductivity.
We thus conjecture that $\kappa = \mathcal O (\lambda')$ in this later case. 
This is in agreement with numerical results in \cite{oga}, where it is suggested that $\kappa$ could even decay as $\ed^{-c/\lambda'}$ for some $c >0$.

In the next {\comm{section}}, we define the model studied in this paper, we state our results and we give some heuristic indications. 
The rest of the paper is then devoted to the proof of Theorems \ref{theo: anharmonic} and \ref{theo: harmonic}. 
Let us already indicate its main steps. 
The principal computation of this article consists in showing that the current due to harmonic interactions between particles $k$ and $k+1$, called $j_{k,har}$, 
can be written as $j_{k,har} = -A_{har} u_{k}$, 
where $u_{k}$ is localized near $k$, and where $A_{har}$ is the generator of the harmonic dynamics. 
This is stated precisely and shown in Section \ref{sec: Poisson equation} ; the proof ultimately rests on localization results first established by Kunz and Souillard (see \cite{kun} or \cite{dam}). 
Once this is seen, general inequalities on Markov processes
allow us to obtain, in Section \ref{sec: upper bound}, the desired upper bound $\kappa = \mathcal O (\lambda)$ in presence of both a noise and non-linear forces.  
The lower bound $\kappa \ge c \lambda$, valid when $\lambda' = 0$, is established by means of a variational formula (see \cite{set}), using a method developed by the first author in \cite{ber}. 
This is carried out in Section \ref{sec: lower bound}.


\section{Model and results}\label{sec: model and results}

\subsection{Model}\label{subsec: model}
We consider a one-dimensional chain of $N$ oscillators, 
so that a state of the system is characterized by a point 
\begin{equation*}
x 
\; = \; 
(q,p) 
\; = \;
(q_1, \dots , q_N,p_1, \dots , p_N) \in \R^{2N}, 
\end{equation*} 
where $q_k$ represents the position of particle $k$, and $p_k$ its momentum. 
The dynamics is made of a hamiltonian part perturbed by a stochastic noise. 

\paragraph{The Hamiltonian.}
The Hamiltonian writes
\begin{align*}
H (q,p) 
\; = &\; 
H_{har}(q,p) + \lambda' H_{anh}(q,p) \\ 
= &\; 
\frac{1}{2}\sum_{k=1}^N \Big( p_k^2 + \nu_k \, q_k^2 + (q_{k+1} - q_k)^2 \Big)
+ \lambda' \sum_{k=1}^N \Big( U (q_k) + V(q_{k+1} - q_{k}) \Big),
\end{align*}
with the following definitions.
\begin{itemize}
\item
The pinning parameters $\nu_k$ are i.i.d. random variables {\comm{whose}} law is independent of $N$.
{\comm{It is assumed that this law has a bounded density and that there exist constants}} $0 < \nu_- < \nu_+ < \infty$ such that 
\begin{equation*}
\Proba (\nu_- \, \le \, \nu_k \, \le \, \nu_+) \; = \; 1.
\end{equation*}
\item
The value of $q_{N+1}$ depends on the boundary conditions (BC).
For fixed BC, we put $q_{N+1} = 0$, while for periodic BC, we put $q_{N+1} = q_1$.
For further use, we also define $q_0 = q_1$ for fixed BC, and $q_0 = q_{N}$ for periodic BC.
\item
We assume $\lambda' \ge 0$.
The potentials $U$ and $V$ are symmetric, meaning that $U(-x) = U(x)$ and $V(-x) = V(x)$ for every $x \in \R$.
They belong to $\mathcal C^{\infty}_{temp}(\R)$, 
the space of infinitely differentiable functions with polynomial growth.
It is moreover assumed that
\begin{equation*}
\int_\R \ed^{-U(x)} \, \dd x \; < \; + \infty 
\quad \text{and} \quad 
\partial_x^2 U (x) \; \ge \; 0,
\end{equation*}
and that there exists $c > 0$ such that 
\begin{equation*}
c \; \le \; 1 + \lambda' \partial_x^2 V (x) \; \le \; c^{-1}.
\end{equation*}
\end{itemize}

For $x=(x_1, \ldots, x_d)\in \RR^d$ and $y=(y_1, \ldots,y_d) \in \R^d$, let $\langle x,y\rangle = x_1 y_1 + \dots + x_d y_d$ be the canonical scalar product of $x$ and $y$.
The harmonic hamiltonian $H_{har}$ can also be written as
\begin{equation*}
H_{har} (q,p) \; = \; \frac{1}{2}\langle p,p \rangle \, + \, \frac{1}{2}\langle q, \Phi q\rangle, 
\end{equation*}
if we introduce the symmetric matrix
$\Phi \in \R^{N\times N}$ of the form $\Phi = -\Delta + W$, where $\Delta$ is the discrete Laplacian, and $W$ a random ``potential".
The precise definition of $\Phi$ depends on the BC: 
\begin{align*}
\Phi_{j,k} \; = &\; (2 + \nu_k) \delta_{j,k} - \delta_{j,k+1} - \delta_{j,k-1}
\qquad \text{(fixed BC)}, \\
\Phi_{j,k} \; = &\; (2 + \nu_k) \delta_{j,k} - \delta_{j,k+1} - \delta_{j,k-1}  -\delta_{j,1}\delta_{k,N} - \delta_{j,N}\delta_{k,1}
\qquad \text{(periodic BC)},
\end{align*}
for $1 \le j,k \le N$.

\paragraph{The dynamics.}
The generator of the hamiltonian part of the dynamics is written as
\begin{equation*}
A \; = A_{har} + \lambda' A_{anh}
\end{equation*}
with
\begin{equation*}
A_{har} 
\; = \; 
\sum_{k=1}^N \big( \partial_{p_k} H_{har}\multiplication \partial_{q_k} - \partial_{q_k} H_{har} \multiplication \partial_{p_k} \big)
\; = \;
\langle p , \nabla_q \rangle - \langle \Phi q , \nabla_p \rangle
\end{equation*}
and
\begin{equation*}
A_{anh}
\; = \; 
-\sum_{k=1}^N 
\partial_{q_k} H_{anh} \multiplication \partial_{p_k}
\; = \; 
- \big( 
\partial_x U(q_k) 
+ \partial_x V (q_k - q_{k-1})
- \partial_x V(q_{k+1} - q_k)
\big) \multiplication \partial_{p_k}. 
\end{equation*}
{\comm{Here, for $x=(x_1, \ldots, x_N) \in \RR^N$, $\nabla_x = (\partial_{x_1}, \ldots, \partial_{x_N})$ .}}The generator of the noise is defined to be
\begin{equation*}
\lambda\, S u 
\; = \; 
\lambda \sum_{k=1}^N \big( u (\dots ,-p_k,\dots ) - u (\dots , p_k, \dots ) \big),
\end{equation*}
with $\lambda \ge \lambda'$.
The generator of the full dynamics is given by 
\begin{equation*}
L = A + \lambda\, S.
\end{equation*}
We denote by $X_{(\lambda,\lambda')}^t(x)$, or simply by $X^t (x)$, the value of the Markov process generated by $L$ at time $t\ge 0$, starting from $x = (q,p)\in\R^{2N}$. 

\paragraph{Expectations.}
Three different expectations will be considered. 
We define
\begin{itemize}
\item
$\mu_T$: the expectation with respect to the Gibbs measure at temperature $T$, 
\item
$\Mean$: the expectation with respect to the realizations of the noise,
\item
$\Mean_\nu$: the expectation with respect to the realizations of the pinnings. 
\end{itemize}
In Section \ref{sec: convergence}, it will sometimes be useful to specify the dependence of the Gibbs measure on the system size $N$ ; 
we then will write it $\mu_T^{(N)}$.


The Gibbs measure $\mu_T$ is explicitly given by 
\begin{equation*}
\mu_T (u) \; = \; \frac{1}{Z_T} \int_{\R^{2N}} u(x) \, \ed^{-H(x)/T} \, \dd x, {\comm{\quad u:\RR^{2N} \to \RR,}}
\end{equation*}
where $Z_T$ is a normalizing factor such that $\mu_T$ is a probability measure on $\R^{2N}$.
We will need some properties of this measure. Let us write
\begin{equation*}
Z_T^{-1} \ed^{- H(x)/T} \;= \; \rho' (p_1) \dots \rho' (p_N) \multiplication \rho'' (q), 
\end{equation*}
with $\rho' (p_k) = \ed^{-p_k^2/2T} / \sqrt{2\pi T}$ for $1 \le k \le N$.

{\comm{When $\lambda' = 0$, the density $\rho''$ is Gaussian:
\begin{equation*}
\rho'' (q) 
\; = \;
(2\pi T)^{-N/2} \multiplication ( \det \Phi )^{1/2} \multiplication \ed^{-\langle q, \Phi q \rangle /2T}. 
\end{equation*}}}
Since $\nu_k \ge \nu_- > 0$, 
it follows from Lemma 1.1 in \cite{bry} that $|(\Phi^{-1})_{i,j}| \le \mathrm C\, \ed^{-c |j-i|}$, for some constants $\mathrm C <+\infty$ and $c > 0$ independent of $N$.
This implies in particular the decay of correlations 
\begin{equation*}
\mu_T (q_i q_j) \; = \; T\, (\Phi^{-1})_{i,j} \; \le \; \mathrm C\, T\, \ed^{-c |j-i|}.
\end{equation*}

When $\lambda' > 0$, the density {\comm{$\rho''$}} is not Gaussian anymore. 
We here impose the extra assumption that $\nu_-$ is large enough. 
In that case, our hypotheses ensure that the conclusions of Theorem 3.1 in \cite{bod} hold: 
there exist constants $\mathrm C < + \infty$ and $c >0$ such that, for every $f,g \in \mathcal C^\infty_{temp} (\R^N)$ satisfying $\mu_T (f) = \mu_T (g) = 0$, 
\begin{equation}\label{decay of correlations Bodineau Helffer}
\big|\mu_T (f\multiplication g) \big|
\; \le \; 
\mathrm C \, \ed^{- c \, d ( \mathrm S(f), \mathrm S (g)) } 
\Big(
\mu_T \big( \big\langle \nabla_q f , \nabla_q f \big\rangle\big)
\multiplication
\mu_T \big( \big\langle \nabla_q g , \nabla_q g \big\rangle\big)
\Big)^{1/2}.
\end{equation} 
Here, $\mathrm S (u) $ is the support of {\comm{the}} function $u$, 
defined as the smallest set of integers such that $u$ can be written as a function of the variables $x_l$ for $l \in \mathrm S (u)$, 
whereas $ d ( \mathrm S(f), \mathrm S (g))$ is the smallest distance between any integer in $\mathrm S(f)$ and any integer in $\mathrm S(g)$.
Using that $\mu_T (q_k) = 0$ for $1 \le k \le N$, 
it is checked from \eqref{decay of correlations Bodineau Helffer} that every function $u\in \mathcal C^\infty_{temp}(\R^N)$
with given support independent of $N$ is such that $\| u \|_{\Lp^1 (\mu_T)}$ is bounded uniformly in $N$.

\paragraph{The current.}
The local energy $e_k$ of atom $k$ is defined as
\begin{equation*}
e_k 
\; = \;
e_{k,har} + \lambda' e_{k,anh}
\end{equation*}
with
\begin{equation*}
e_{k,har}
\; = \; 
\frac{p_k^2}{2} + \nu_k \frac{q_k^2}{2} + \frac{1}{4} (q_k - q_{k-1})^2 + \frac{1}{4} (q_{k+1} - q_k)^2
\qquad \text{for} \qquad 2 \le k \le N-1,
\end{equation*}
and
\begin{equation*}
e_{k,anh}
\; = \; 
U (q_k) + \frac{V(q_k - q_{k-1})}{2} + \frac{V(q_{k+1} - q_k)}{2}
\qquad \text{for} \qquad 2 \le k \le N-1.
\end{equation*}
For periodic B.C., these expressions are still valid when $k=1$ or $k=N$. 
For fixed B.C. instead, all the terms involving the differences $(q_0 - q_1)$ or $(q_{N+1} - q_N)$ in the previous expressions {\comm{have}} to be multiplied by $2$. 
These definitions ensure that the total energy $H$ is the sum of the local energies. 

The definition of the dynamics implies that 
\begin{equation*}
\dd e_k 
\; = \;
\big( j_{k-1} - j_{k} \big) \, \dd t
\end{equation*}
for local currents 
\begin{equation*}
j_k \; = \; j_{k,har} + \lambda' j_{k,anh}
\end{equation*}
defined as follows for $0 \le k \le N$.
First, for $1 \le k \le N-1$,
\begin{equation}\label{local current}
j_{k,har} 
\; = \;
\frac{1}{2} (p_k + p_{k+1}) (q_{k} - q_{k+1})
\qquad
\text{and}
\qquad
j_{k,anh}
\; = \;
\frac{1}{2} (p_k + p_{k+1}) \,
\partial_x V (q_{k} - q_{k+1}).
\end{equation}
Next, $j_{0,1} = j_{N,N+1} = 0$ for fixed B.C.
Finally, $j_{0}$ and $j_{N}$ are still given by \eqref{local current} for periodic B.C., with the conventions $p_0 = p_N$ and $p_{N+1}=p_1$.
The total current and the rescaled total current are then defined by 
\begin{align}
J_N 
\; = &\; 
J_{N,har} + \lambda' J_{N,anh}
\; = \; 
\sum_{k=1}^N j_{k,har} + \lambda' \sum_{k=1}^{N} j_{k,anh}
\label{total current} \\
\mathcal J_N 
\; = &\;
\mathcal J_{N,har} + \lambda' \mathcal J_{N,anh}
\; = \; 
\frac{ J_{N,har}}{\sqrt N} + \lambda' \frac{ J_{N,anh}}{\sqrt N}.
\label{rescaled current}
\end{align}

\subsection{Results}
For a given realization of the pinnings, the (Green-Kubo) conductivity $\kappa = \kappa (\lambda, \lambda')$ of the chain is defined as 
\begin{equation}\label{Green Kubo}
\kappa (\lambda,\lambda') 
\; = \; 
\frac{1}{T^2}\lim_{t\rightarrow \infty}\lim_{N\rightarrow \infty} \kappa_{t,N}(\lambda,\lambda')
\; =\;
\frac{1}{T^2}\lim_{t\rightarrow \infty}\lim_{N\rightarrow \infty}
\mu_T\Mean \Bigg( \frac{1}{\sqrt t}\int_0^t \mathcal J_N \circ X_{(\lambda,\lambda')}^s\, \dd s \Bigg)^2
\end{equation}
if this limit exists. 
The choice of the boundary conditions is expected to play no role in this formula since the {\comm{volume size}} $N$ is sent to infinity for fixed time.
The disorder averaged conductivity is defined by replacing $ \mu_T\Mean$ by $\Mean_\nu \mu_T\Mean$ in \eqref{Green Kubo}. 
By ergodicity, the conductivity and the disorder averaged conductivity are expected to coincide for almost all realization of the pinnings (see \cite{ber}).
The dependence of $\kappa (\lambda , \lambda')$ on the temperature $T$ will not be analysed in this work, 
so that we can consider $T$ as a fixed given parameter.

We first obtain an upper bound on the disorder averaged conductivity.
\begin{Theorem}\label{theo: anharmonic}
Let $0 \le \lambda' \le \lambda$. 
With the assumptions introduced up to here, if $\nu_-$ is large enough,
and for fixed boundary conditions, 
\begin{equation}\label{upper bound on Green Kubo}
\frac{1}{T^2}\limsup_{t\rightarrow \infty}\limsup_{N\rightarrow \infty}
\Mean_\nu \mu_T \Mean \Bigg( \frac{1}{\sqrt t}\int_0^t \mathcal J_N \circ X^s_{(\lambda,\lambda')}\, \dd s \Bigg)^2 \; = \; \mathcal O(\lambda)
\qquad \text{as} \qquad \lambda \rightarrow 0.
\end{equation} 
\end{Theorem}
\Remarks
1.
When $\lambda = 0$, the proof (see Section \ref{sec: upper bound}) actually shows that 
\begin{equation*}
\frac{1}{T^2} \limsup_{N\rightarrow \infty}\Mean_\nu\mu_T \Bigg( \frac{1}{\sqrt t}\int_0^t \mathcal J_N \circ X^s_{(0,0)}\, \dd s \Bigg)^2
\; = \; 
\mathcal O \big( t^{-1} \big)
\qquad \text{as} \qquad t\rightarrow \infty.
\end{equation*}
This bound had apparently never been published before. 
It says that the unperturbed chain behaves like a perfect insulator: 
the current integrated over arbitrarily long times remains bounded in $\Lp^2(\Mean_\nu \mu_T)$.

\noindent
2.
The proof (see Section \ref{sec: upper bound}) shares some common features with a method used in \cite{liv} 
to obtain a weak coupling limit for noisy hamiltonian systems.
In our case, we may indeed see the eigenmodes of the unperturbed system as weakly coupled by the noise and the anharmonic potentials. 

\noindent
3. 
The choice of fixed boundary conditions just turns out to be {\comm{more}} convenient for technical reasons (see Section \ref{sec: Poisson equation}).

\noindent
4. 
The hypothesis that $\nu_-$ is large enough is only used to ensure the exponential decay of correlations of the Gibbs measure when $\lambda' >0$. 

Next, in the absence of anharmonicity ($\lambda' = 0$), results become more refined.
\begin{Theorem}\label{theo: harmonic}
Let $\lambda > 0$, let $\lambda' = 0$, and let us assume that hypotheses introduced up to here hold.
For almost all realizations of the pinnings, 
the Green-Kubo conductivity \eqref{Green Kubo} of the chain is well defined, and in fact
\begin{equation}\label{convergence of Green Kubo}
\kappa (\lambda,0) 
\; = \; 
\frac{1}{T^2}\lim_{t\rightarrow \infty}\lim_{N\rightarrow \infty} \Mean_\nu\mu_T\Mean \Bigg( \frac{1}{\sqrt t}\int_0^t \mathcal J_N \circ X_{(\lambda,0)}^s\, \dd s \Bigg)^2,
\end{equation}
this last limit being independent of the choice of boundary conditions (fixed or periodic).
Moreover, there exists a constant $c > 0$ such that, for every $\lambda \in ]0,1[$,
\begin{equation}\label{upper and lower bound on Green Kubo}
c \lambda \; \le \; \kappa (\lambda, 0) \; \le \; c^{-1} \lambda.
\end{equation}
\end{Theorem}

{\comm{The rest of this}} article is devoted to the proof of these theorems, which is constructed as follows.

\ProofOf{of Theorems \ref{theo: anharmonic} and \ref{theo: harmonic}}
The upper bound \eqref{upper bound on Green Kubo} is derived in Section \ref{sec: upper bound}, 
assuming that Lemma \ref{lemm: Poisson equation} holds. 
This lemma is stated and shown in Section \ref{sec: Poisson equation} ; 
it encapsulates the informations we need about the localization of the eigenmodes of the unperturbed system ($\lambda = \lambda' = 0$).
The existence of $\kappa (\lambda , 0)$ for almost every realization of the pinnings, together with \eqref{convergence of Green Kubo},
are shown in Section \ref{sec: convergence}.
Finally, a lower bound on the conductivity when $\lambda' = 0$ is obtained in Section \ref{sec: lower bound}.
This shows \eqref{upper and lower bound on Green Kubo}. 
$\square$

\subsection{Heuristic comments}\label{subsec: heuristic}
We would like to give here some intuition on the conductivity of disordered harmonic chains perturbed by a weak noise only, 
so with $\lambda >0$ small and $\lambda' = 0$.
We will develop in a more probabilistic way some ideas from \cite{dha}.
Our results cover the case where the pinning parameters $\nu_k$ are bounded from below by a positive constant, but 
it could be obviously desirable to understand the unpinned chain as well, in which case randomness has to be putted on the value of the masses.
We handle here both cases.

Let us first assume that $\nu_k \ge c$ for some $c > 0$, and let us consider a typical realization of the pinnings. 
In the absence of noise ($\lambda = 0$), 
the dynamics of the chain is actually equivalent to that of $N$ independent one-dimensional harmonic oscillators, called eigenmodes
(see Subsection \ref{subsec: eigenmode expansion} and formulas (\ref{eigenmode solution q}-\ref{eigenmode solution p}) in particular).
Since the chain is pinned at each site, the eigenfrequencies of these modes are uniformly bounded away from zero. 
As a result, all modes are expected to be exponentially localized. 
We can thus naively think that, to each particle, is associated a mode localized near the equilibrium position of this particle. 

When the noise is turned on ($\lambda >0$), energy starts being exchanged between near modes.
Let us assume that, initially, energy is distributed uniformly between all the modes, except around the origin, where some more energy is added. 
We expect this extra amount of energy to diffuse with time, with a variance proportional to $\kappa (0,\lambda) \multiplication t$ at time $t$.
Since flips of velocity occur at random times and with rate $\lambda$, 
we could compare the location of this extra energy at time $t$ to the position of a standard random walk after $n = \lambda t$ steps. 
Therefore, denoting by $\delta_k$ the increments of this walk, we find that
\begin{equation*}
\kappa (\lambda , 0) 
\; \sim \; 
\Big\langle \Big( \frac{1}{\sqrt{t}} \sum_{k=1}^n \delta_k \Big)^2 \Big\rangle 
\; \sim \; \lambda. 
\end{equation*}
This intuitive picture will only be partially justified, 
as explained in the remark after the proof of Theorem \ref{theo: anharmonic} in Section \ref{sec: upper bound}. 

Let us now consider the unpinned chain.
So we put $\nu_k = 0$ and we change $p_k^2$ by $p_k^2 /m_k$ in the Hamiltonian, where {\comm{the masses}} $m_k$ are i.i.d. positive random variables.
We consider a typical realization of the masses. 
In contrast with the pinned chain, the eigenfrequencies of the modes are now distributed in an interval of the form $[0,c]$, for some $c >0$.
This has an important consequence on the localization of the modes.
It is indeed expected that the localization length $l$ of a mode and its eigenfrequency $\omega$ are related through the formula $l \sim 1/ \omega^2$.

Here again, the noise induces exchange of energy between modes, 
and we still would like to compare $\kappa (\lambda , 0) \multiplication t$ with the variance of a centered random walk with increments $\delta_k$.
However, due to the unlocalized low modes, $\delta_k$ can now take larger values than in the pinned case. 
Assuming that the eigenfrequencies are uniformly distributed in $[0,c]$, we guess that, for large $a$, 
\begin{equation*}
\Proba (|\delta_k| \ge a) 
\; \sim \;
\Proba (1/\omega^2 \ge a)
\; \sim \; 
1/\sqrt a.
\end{equation*}
This however neglects a fact. 
Since energy does not travel faster than ballistically, 
and since successive flips of the velocity are spaced by time intervals of order $1/\lambda$, 
it is reasonable to introduce the cut-off $\Proba (|\delta_k| > 1/ \lambda) = 0$.
With this distribution for $|\delta_k|$, and with $n = \lambda t$, we now find
\begin{equation*}
\kappa (\lambda , 0) 
\; \sim \; 
\Big\langle \Big( \frac{1}{\sqrt{t}} \sum_{k=1}^n \delta_k \Big)^2 \Big\rangle 
\; \sim \; \lambda^{-1/2}. 
\end{equation*}
This scaling is numerically observed in \cite{dha}. 
The arguments leading to this conclusion are very approximative however, and it should be desirable to analyse this case rigorously as well.

\section{Upper bound on the conductivity}\label{sec: upper bound}
We here proceed to the proof of Theorem \ref{theo: anharmonic}.
We assume that Lemma \ref{lemm: Poisson equation} in Section \ref{sec: Poisson equation} holds: 
there exists a sequence $(u_N)_{N\ge 1} \subset \Lp^2 (\Mean_\nu \mu_T)$ such that 
$- A_{har} u_N = \mathcal J_{N,har}$, 
and that $(u_N)_{N \ge 1}$ and $(A_{anh} u_N)_{N\ge 1}$ are both bounded sequences in $\Lp^2 (\Mean_\nu \mu_T)$. 
Moreover $u_N$ is of the form $u_N (q,p) = \langle q, \alpha_N q \rangle + \langle p, \gamma_N p \rangle + c_N$, 
where $\alpha_N ,\gamma_N \in \R^{N\times N}$ are symmetric matrices, and where $c_N \in \R$. 

\ProofOf{of \eqref{upper bound on Green Kubo}} 
Let $0 \le \lambda' \le \lambda$, and let $u_N$ be the sequence obtained by Lemma \ref{lemm: Poisson equation} in Section \ref{sec: Poisson equation}.
Before starting, let us observe that, due to the special form of the function $u_N$, we may write
\begin{equation}\label{def A anh u 1}
A_{anh} u_N \; = \; \sum_{l=1}^N \phi_l (q) \, p_l, 
\end{equation}
with
\begin{equation}\label{def A anh u 2}
\phi_l (q) \; = \; 2 \sum_{k=1}^N \gamma_{k,l} \multiplication
\Big(
\partial_x V (q_{k+1} - q_k) - \partial_x V (q_k - q_{k-1}) - \partial_x U (q_k)
\Big), 
\end{equation}
where $(\gamma_{k,l})_{1 \le k,l\le N}$ are the entries of $\gamma_N$.
It follows in particular that
\begin{equation}\label{Aanh u is in the image of S}
A_{anh} u_N \; = \; \frac{1}{2} (-S) A_{anh} u_N.
\end{equation}

Now, since $\mathcal J_N = \mathcal J_{N,har} + \lambda' \mathcal J_{N,anh}$, we find using Cauchy-Schwarz inequality that
\begin{multline*}
\Mean_\nu \mu_T \Mean \left( \frac{1}{\sqrt t} \int_0^t \mathcal J_N \circ X^s \, \dd s \right)^2
\; \le \;
\Mean_\nu \mu_T \Mean \left( \frac{1}{\sqrt t} \int_0^t \mathcal J_{N,har} \circ X^s \, \dd s \right)^2
+
(\lambda')^2 \Mean_\nu \mu_T \Mean \left( \frac{1}{\sqrt t} \int_0^t \mathcal J_{N,anh} \circ X_s \, \dd s \right)^2 \\
+
2 \lambda'\,
\left(
\Mean_\nu \mu_T\Mean \left( \frac{1}{\sqrt t} \int_0^t \mathcal J_{N,har} \circ X^s \, \dd s \right)^2 
\multiplication
\Mean_\nu \mu_T\Mean \left( \frac{1}{\sqrt t} \int_1^t \mathcal J_{N,anh} \circ X^s \, \dd s \right)^2 
\right)^{1/2}.
\end{multline*}
Since $\mathcal J_{N,anh} = \frac{1}{2}(-S)\mathcal J_{N,anh}$, a classical bound {\comm{(\cite{kip}, Appendix 1, Proposition 6.1)}} furnishes 
\begin{equation*}
\Mean_\nu \mu_T\Mean \left( \frac{1}{\sqrt t} \int_0^t \mathcal J_{N,anh} \circ X^s \, \dd s \right)^2 
\; \le \; 
\mathrm C \, \Mean_\nu\mu_T \big( \mathcal J_{N,anh} \multiplication (-\lambda S)^{-1} \mathcal J_{N,anh} \big)
\; \le \; 
\frac{\mathrm C}{2\lambda} \, \Mean_\nu\mu_T \big( \mathcal J_{N,anh}^2 \big)
\end{equation*}
where $\mathrm C < +\infty$ is a universal constant.
By \eqref{decay of correlations Bodineau Helffer}, $\Mean_\nu\mu_T \big( \mathcal J_{N,anh}^2 \big)$ is  {\comm{uniformly bounded}} in $N$.
Therefore
\begin{equation*}
\limsup_{t\rightarrow \infty} \limsup_{N\rightarrow \infty}
\Mean_\nu \mu_T\Mean \left( \frac{1}{\sqrt t} \int_0^t \mathcal J_{N,anh} \circ X^s \, \dd s \right)^2 
\; = \; 
\mathcal O (\lambda^{-1}).
\end{equation*}
It suffices thus to establish that
\begin{equation*}
\limsup_{t\rightarrow \infty} \limsup_{N\rightarrow \infty} \Mean_\nu \mu_T \Mean \left( \frac{1}{\sqrt t} \int_0^t \mathcal J_{N,har} \circ X^s \, \dd s \right)^2 
\; = \; \mathcal O (\lambda) .
\end{equation*}
We write
\begin{equation*}
\mathcal J_{N,har} 
\; = \; 
-A_{har} u_N 
\; = \; 
-L u_N + \lambda' A_{anh} u_N + \lambda S u_N
\; = \; 
- L u_N + \lambda \, S \Big( Id - \frac{\lambda'}{2\lambda} A_{anh} \Big) u_N,
\end{equation*}
where the second equality is obtained by means of \eqref{Aanh u is in the image of S}.
Therefore
\begin{align} 
\frac{1}{\sqrt t} \int_0^t \mathcal J_{N,har} \circ X^s \, \dd s
\; = &\; 
\frac{-1}{\sqrt t} \int_0^t L u_N \circ X^s \, \dd s
+
\frac{\lambda}{\sqrt t} \int_0^t S \left( Id - \frac{\lambda'}{2\lambda }A_{anh} \right) u_N \circ X^s \, \dd s \nonumber \\
=&\;
\frac{1}{\sqrt t} \mathcal M_t - \frac{u \circ X^t - u}{\sqrt t}
+
\frac{\lambda}{\sqrt t} \int_0^t S \left( Id - \frac{\lambda'}{2\lambda}A_{anh} \right) u_N \circ X^s \, \dd s, \label{decomposition of the process}
\end{align}
where $\mathcal M_t$ is a martingale given by
\begin{equation*}
\mathcal M_t 
\; = \; 
\int_0^t \sum_{j=1}^N S_j u_N \circ X_s \, (\dd \mathrm N_s^j - \lambda \dd s),
\end{equation*}
with $\mathrm N_s^j$ the Poisson process that flips the momentum of particle $j$. 

It now suffices to establish that the three terms in the right hand side of \eqref{decomposition of the process} are $\mathcal O (\lambda)$ in $\Lp^2 (\Mean_\nu\Mean_T)$.
Let us first show that $\mu_T ( u_N \multiplication (-S) u_N ) \le 4 \| u_N \|_{\Lp^2 (\mu_T)}^2$.
Writing
\begin{equation*}
u_N = u_{N}^{p,p,0} +u_{N}^{p,p,1} + u_{N}^{q,q} + c_N
\end{equation*}
with
\begin{equation*}
u_{N}^{p,p,0} =\sum_{i \ne j} \gamma_{i,j} p_{i}p_j , 
\quad u_N^{p,p,1}= \sum_i \gamma_{i,i} p_i^2,
\quad u_{N}^{q,q} = \langle q, \alpha_N q \rangle,
\end{equation*}
we get indeed 
\begin{equation*}
\mu_T ( u_N \multiplication (-S)u_N )
\; = \; 
\mu_T ( u_N \multiplication (-S)u_N^{p,p,0} )
\; = \;
4 \, \mu_T ( u_N^{p,p,0} \multiplication u_N^{p,p,0} )
\end{equation*}
and
\begin{equation*}
\mu_T ( u_N \multiplication u_N ) 
\; = \; 
\mu_T \big( u_N^{p,p,0} \multiplication u_N^{p,p,0} \big) + 
\mu_T \big( (u_{N}^{p,p,1} + u_{N}^{q,q} + c_N )^2 \big)
+ 
2\, \mu_T \big( (u_{N}^{p,p,1} + u_{N}^{q,q} + c_N ) \multiplication u_N^{p,p,0} \big) .
\end{equation*} 
The claim follows since $\mu_T \big( (u_{N}^{p,p,1} + u_{N}^{q,q} + c_N ) \multiplication u_N^{p,p,0} \big) = 0$.

So first, 
\begin{equation*}
\mu_T\Mean \left( \frac{1}{\sqrt t} \mathcal M_t \right)^2 
\; = \; 
2 \lambda \mu_T\big(  u_N \multiplication (-S) u_N \big)
\; \le \; 8 \lambda \, \|u_N \|^2_{\Lp^2 (\mu_T)}. 
\end{equation*}
Next, 
\begin{equation*}
\mu_T\Mean \left( \frac{u \circ X_t - u}{\sqrt t} \right)^2
\; \le \; 
\frac{2}{t} \,
\| u_N\|^2_{\Lp^2 (\mu_T)} .
\end{equation*}
Finally, by a classical bound {\comm{(\cite{kip}, Appendix 1, Proposition 6.1)}}, 
\begin{align*}
\mu_T\Mean \left( \frac{\lambda}{\sqrt t} \int_0^t S \left(Id - \frac{\lambda'}{2\lambda}A_{anh} \right) u_N \circ X_s \, \dd s \right)^2 
\; \le &\; 
\mathrm C \, \lambda^2 \mu_T\left( S \left(Id - \frac{\lambda'}{2\lambda}A_{anh} \right) u_N \multiplication
(-\lambda S)^{-1} S \left(Id - \frac{\lambda'}{2\lambda}A_{anh} \right) u_N \right) \\
\; = &\; 
\mathrm C \, \lambda \mu_T\left( \left(Id - \frac{\lambda'}{2\lambda}A_{anh} \right) u_N \multiplication (-S) \left(Id - \frac{\lambda'}{2\lambda}A_{anh} \right) u_N \right) \\
\; = &\; \mathrm C \, \lambda \left( \mu_T \big( u_N \multiplication (-S) u_N \big) + \frac{1}{2} \left(\frac{\lambda'}{\lambda}\right)^2 \| A_{anh} u _N \|^2_{\Lp^2 (\mu_T)} \right) \\
\; \le &\; \mathrm C \, \lambda \left( \|u_N \|^2_{\Lp^2 (\mu_T)} + \| A_{anh} u _N \|^2_{\Lp^2 (\mu_T)} \right) 
\end{align*}
where \eqref{Aanh u is in the image of S} and 
\begin{equation*}
\mu_T ( u_N, A_{anh} u_N )
\; =\; 
\mu_T\Big( \big( \langle q, \alpha_N q \rangle + \langle p, \gamma_N p \rangle + c_N \big) \multiplication \sum_{l=1}^N \phi_l (q) \, p_l \Big)
\; = \; 
0
\end{equation*}
have been used to get the second equality.
Taking the expectation over the pinnings, the proof is completed since $(u_N)_N$ and $(A_{anh} u_N)_N$ are bounded sequences in $\Lp^{2}(\Mean_\nu \mu_T)$.
$\square$

\Remark
When $\lambda' = 0$, formula \eqref{decomposition of the process} becomes 
\begin{equation}\label{decomposition of the process when no anharmonicity} 
\frac{1}{\sqrt t} \int_0^t \mathcal J_N \circ X^s \, \dd s
\; = \;
\frac{1}{\sqrt t} \mathcal M_t - \frac{u_N \circ X^t - u_N}{\sqrt t}
+
\frac{\lambda}{\sqrt t} \int_0^t S u_N \circ X^s \, \dd s.
\end{equation}
Now, since $Su_N = -4 \sum_{1 \le k\ne l \le N}\gamma_{k,l}\, p_k p_l$, it is computed that 
\begin{equation*}
\int_0^t \mathcal J_N \circ X_{t-s} \, \dd s 
\; = \; 
- \int_0^t \mathcal J_N \circ X_s \, \dd s 
\quad \text{and} \quad
\int_0^t S u_N \circ X_{t-s} \, \dd s 
\; = \; 
\int_0^t S u_N \circ X_s \, \dd s. 
\end{equation*}
The measure on the paths being invariant under time reversal, it thus holds that 
\begin{equation*}
\mu_T\Mean \Bigg( \int_0^t \mathcal J_N \circ X_s \, \dd s \multiplication \int_0^t S u_N \circ X_s \, \dd s \Bigg)
\; = \; 0.
\end{equation*}
We therefore deduce from \eqref{decomposition of the process when no anharmonicity} that 
\begin{equation*}
\mu_T\Mean \Bigg( \frac{1}{\sqrt t}\int_0^t \mathcal J_N (s)\, \dd s \Bigg)^2
\; = \; 
\mu_T\Mean \Big( \frac{1}{\sqrt t} \mathcal M_t \Big)^2
- \; 
\mu_T\Mean \Bigg( \frac{\lambda}{\sqrt t}\int_0^t Su_N \circ X_s \, \dd s \Bigg)^2 
\; + \; r(t)
\end{equation*}
where $r (t)$ is quantity that vanishes in the limit $t\rightarrow \infty$.
We see thus that our proof does not completely justify the heuristic developed in Subsection \ref{subsec: heuristic}, 
due to the second term in the right hand side of this last equation. 
As explained after the statement of Lemma \ref{lemm: Poisson equation} below, the sequence $u_N$ should not be unique. 
It could be that a good choice of sequence $u_N$ {\comm{makes this second term of order $\mathcal O (\lambda^2)$}}.

\section{Poisson equation for the unperturbed dynamics}\label{sec: Poisson equation}
In this section, we state and prove the following lemma.
Fixed BC are assumed for the whole section. 
\begin{Lemma}\label{lemm: Poisson equation}
Let $\lambda' \ge 0$, and assume fixed boundary conditions.
For every $N \ge 1$, and for almost every realization of the pinnings, there exist a function $u_N$ of the form
\begin{equation*}
u_N (q,p) \; = \; \langle q, \alpha_N q \rangle + \langle p , \gamma_N p \rangle + c_N,
\end{equation*}
where $\alpha_N,\gamma_N\in \R^{N\times N}$ are symmetric matrices and where $c_N \in \R$,
such that 
\begin{equation}\label{Poisson equation}
- A_{har} u_N \; = \; \mathcal J_{N,har}. 
\end{equation}
Moreover, the functions $u_N$ can be taken so that 
\begin{equation*}
(u_N)_{N \ge 1} 
\quad \text{and} \quad
(A_{anh} u_N)_{N\ge 1}
\quad \text{are bounded sequences in} \quad 
\Lp^{2}(\Mean_\nu \mu_T).
\end{equation*}
\end{Lemma}
\Remarks
1. The parameter $\lambda'$ only plays a role through the definition of the measure $\mu_T$.

\noindent
2. For a given value of $N$ and for almost every realization of the pinnings, 
the unperturbed dynamics is integrable, meaning here that it can be decomposed {\comm{into}} $N$ ergodic components, 
each of them corresponding to the motion of a single one-dimensional harmonic oscillator
(see Subsection \ref{subsec: eigenmode expansion} and (\ref{eigenmode solution q}-\ref{eigenmode solution p}) in particular).
This has two implications. 
First, since \eqref{Poisson equation} admits a solution, we conclude that 
the current $\mathcal J_N$ is of mean zero with respect to the microcanonical measures of each ergodic component of the dynamics. 
Next, the solution $u_N$ is not unique since every function $f$ constant on the ergodic components of the dynamics satisfies $-A_{har} f = 0$.

\ProofOf{of Lemma \ref{lemm: Poisson equation}} 
To simplify notations, we will generally not write the dependence on $N$ explicitly.
The proof is made of several steps.

\subsection{Identifying $(u_N)_{N\ge 1}$: eigenmode expansion}\label{subsec: eigenmode expansion}
Let $z > 0$ and let $1 \le l,m \le N$. 
Let us consider the equation 
\begin{equation*}
(z - A_{har}) v_{l,m,z} \; = \; q_l p_m.
\end{equation*}
The solution $v_{l,m,z}$ exists and is unique. 
It is given by 
\begin{equation}\label{basic expression u l m z N}
v_{l,m,z} (x) \; = \; \int_0^{\infty} \ed^{-zs} \left[ q_l \circ X^s_{(0,0)} (x) \multiplication p_m \circ X^s_{(0,0)} (x) \right]  \, \dd s.
\end{equation}
We will analyse $v_{l,m,z}$ to obtain the sequence $u_N$.
Although we assumed fixed BC, all the results of this subsection apply for periodic BC as well.

\paragraph{Solutions to Hamilton's equations.}
The matrix $\Phi$ is a real symmetric positive definite matrix in $\R^{N\times N}$, 
and there exist thus an orthonormal basis $(\xi^k)_{1 \le k \le N}$ of $\R^N$, and a sequence of positive real numbers $(\omega_k^2)_{1 \le k \le N}$, such that 
\begin{equation*}
\Phi \xi^k \; = \, \omega^2_k\, \xi^k. 
\end{equation*}
It may be checked that
\begin{equation}\label{bounds eigenvalues}
\min \{\nu_j : 1 \le j \le N \} \; \le \; \omega_k^2 \; \le \; \max \{\nu_j : 1 \le j \le N \} + 4
\end{equation}
for $1 \le k \le N$.
According to Proposition II.1 in \cite{kun}, for almost all realization of the pinnings, none of the eigenvalue is degenerate: 
\begin{equation}\label{non degeneracy eigenfrequencies}
\omega_j \ne \omega_k \quad \text{if} \quad j\ne k, \qquad 1 \le j,k \le N. 
\end{equation}
In the sequel, we will assume that \eqref{non degeneracy eigenfrequencies} holds.

When $\lambda = \lambda' = 0$, Hamilton's equations write 
\begin{equation*}
\dd q \; = \; p \, \dd t, 
\qquad
\dd p \; = \; - \Phi q\, \dd t.
\end{equation*}
For initial conditions $(q,p)$, the solutions write 
\begin{align}
q(t) \; = &\; \sum_{k=1}^N \Big( \langle q , \xi^k \rangle \cos \omega_k t + \frac{1}{\omega_k} \langle p , \xi^k \rangle \sin \omega_k t \Big) \, \xi^k,
\label{eigenmode solution q}\\
p(t) \; = &\; \sum_{k=1}^N \Big(-\omega_k \langle q , \xi^k \rangle \sin \omega_k t + \langle p , \xi^k \rangle \cos \omega_k t \Big) \, \xi^k.
\label{eigenmode solution p}
\end{align} 

\paragraph{An expression for $v_{l,m,z}$.}
To determine $v_{l,m,z}$, we just need to insert the solutions (\ref{eigenmode solution q}-\ref{eigenmode solution p}) {\comm{into}} the definition \eqref{basic expression u l m z N}, 
and then compute the integral, which is a sum of Laplace transforms of sines and cosines: 
\begin{align}
v_{l,m,z} (q,p) \; = &\;
\frac{1}{4} \sum_{k=1}^N \langle l , \xi^k \rangle \langle m , \xi^k \rangle
\Big( - \langle q , \xi^k \rangle^2 + \frac{1}{\omega_k^2} \langle p , \xi^k \rangle^2 \Big) \nonumber \\
& + \sum_{1\le j\ne k \le N} \langle l , \xi^j \rangle \langle m , \xi^k \rangle
\Big( 
\frac{\omega_k^2}{\omega_j^2 - \omega_k^2} \langle q , \xi^j \rangle \langle q , \xi^k \rangle 
+ \frac{1}{\omega_j^2 - \omega_k^2} \langle p , \xi^j \rangle \langle p , \xi^k \rangle 
\Big) \nonumber \\
& + \mathcal O (z), \label{expression for u l m z N en modes propres}
\end{align}
where $\langle j ,\xi^k \rangle$ denotes the $j^{\text{th}}$ component of the vector $\xi^k$, 
and where the rest term $\mathcal O (z)$ is a polynomial of the form 
$\langle q, \tilde\alpha_z q \rangle + \langle q, \tilde \beta_z p \rangle + \langle p, \tilde\gamma_z p \rangle$,
where $\tilde \alpha_z$ and $\tilde \gamma_z$ can be taken to be symmetric. 
We define 
\begin{equation*}
v_{l,m} \; = \; \lim_{z\rightarrow 0} v_{l,m,z}.
\end{equation*}
It is observed that $v_{l,m}$ is of the form $\langle q, \tilde\alpha q \rangle + \langle p, \tilde\gamma p \rangle$
where $\tilde \alpha$ and $\tilde \gamma$ can be taken to be symmetric. 

\paragraph{Defining the solution $u_N$.}
For fixed BC, the total current is given by
\begin{equation*}
J_N \; = \; \frac{1}{2}(q_1 p_1 - q_N p_N) \, + \, \frac{1}{2} \sum_{k=1}^{N-1} \big(q_k p_{k+1} - q_{k+1} p_k \big) 
\end{equation*}
Setting
\begin{equation}\label{def of w l}
w_{l} \; = \; v_{l,l-1} - v_{l-1,l} \; - \; \mu_T (v_{l,l-1} - v_{l-1,l})
\end{equation}
for $2 \le l \le N$
and 
\begin{equation}\label{def of w 1}
w_1 \; = \; v_{N,N} - v_{1,1} \; - \; \mu_T (v_{N,N} - v_{1,1}),
\end{equation}
we define
\begin{equation*}
u_N \; = \; \frac{-1}{2 \sqrt N} \sum_{k=1}^{N} w_{l}.
\end{equation*}
The function $u_N$ is of the form $u_N = \langle q, \alpha_N q \rangle + \langle p, \gamma_N p \rangle + c_N$, 
where $\alpha_N$ and $\gamma_N$ are symmetric matrices, and where $c_N\in\R$. 

Let us show that $u_N$ solves $-A_{anh} u_N = \mathcal J_N$.
We may assume that $c_N = 0$ without loss of generality.
The current $\mathcal J_N$ can be written as $\mathcal J_N = \langle q, B p \rangle$. 
The function $u_N$ has been obtained as the limit as $z \rightarrow 0$ of the function $u_{z}$ of the form
$u_{z} = \langle q, \alpha_z q \rangle + \langle q, \beta_z p \rangle + \langle p, \gamma_z p \rangle$
which solves $(z - A_{har}) u_{z} = \mathcal J_N$, and with $\alpha_z$ and $\gamma_z$ symmetric matrices.
Since
\begin{equation*}
(z - A_{har}) u_{z} 
\; = \; 
\langle q, (z \alpha_z + \beta_z \Phi) q \rangle
+
\big\langle q, \big( z \beta_z - 2(\alpha_z - \Phi \gamma_z) \big) p \big\rangle
+
\langle p, (z - \beta_z) p \rangle,
\end{equation*}
it holds that {\footnote{{\comm{Here and in the following $M^\dagger$ denotes the transpose matrix of the matrix $M$.}} }}
\begin{equation*}
z \alpha_z + \frac{1}{2} (\beta_z \Phi + \Phi \beta_z^\dagger ) \; = \; 0, 
\quad
z \beta_z - 2(\alpha_z - \Phi \gamma_z) \; = \; B,
\quad
z - \frac{1}{2} (\beta_z + \beta_z^\dagger) \; = \; 0.
\end{equation*}
We know that $(\alpha_z,\beta_z,\gamma_z) \rightarrow (\alpha, 0 , \gamma)$ as $z\rightarrow 0$, with $\alpha$ and $\gamma$ symmetric, so that 
$-2 (\alpha - \Phi \gamma) = B$.
Taking into account that $\alpha$ and $\gamma$ are symmetric, we deduce that
\begin{equation}\label{alpha N and gamma N}
\Phi \gamma - \gamma \Phi \; = \; \frac{1}{2} (B - B^\dagger), 
\qquad
2 \alpha \; = \; 2 \Phi \gamma - B. 
\end{equation}
It is checked that, if two symmetric matrices $\alpha$ and $\gamma$ satisfy these relations, 
then $u_N = \langle q, \alpha q \rangle + \langle p, \gamma p \rangle$ solves the equation $- A_{har} u_N = \mathcal J_N$.

\subsection{A new expression for $w_{l}$}
For $1 \le l \le N$, the function $w_l$ defined by \eqref{def of w l} or \eqref{def of w 1} can be written as 
\begin{equation*}
w_l \; = \; \langle q , \alpha (l) q \rangle + \langle p , \gamma (l) p\rangle + c(l), 
\end{equation*}
where $\alpha (l)$ and $\gamma (l)$ are symmetric matrices, and where $c(l)\in \R$. 
A relation similar to \eqref{alpha N and gamma N} is satisfied: 
with the definitions
\begin{equation*}
\big( B(l) \big)_{m,n} \; = \; \delta_{l,l-1} (m,n) - \delta_{l-1,l} (m,n) \quad (2 \le l \le N)
\quad \text{and} \quad
\big( B(1) \big)_{m,n} \; = \; \delta_{N,N}(m,n) - \delta_{1,1} (m,n), 
\end{equation*}
for $1 \le m,n \le N$,
we write
\begin{equation}\label{link between alpha and gamma}
2 \alpha (l) \; = \; 2 \Phi \gamma (l) - B (l), 
\end{equation}
for $1 \le l \le N$.
Therefore the knowledge of the matrices $\gamma$ implies that of the matrices $\alpha$.

An expression for the matrices $\gamma (l)$ can be recovered from \eqref{expression for u l m z N en modes propres} with $z=0$.
We will now work this out in order to obtain a more tractable formula.
We show here that, for $2 \le l \le N$, 
\begin{align}
\gamma_{s,s} (l)
\; = &\; 
- \sum_{j=l}^N \sum_{k=1}^N \langle s , \xi^k \rangle^2 \langle j , \xi^k \rangle^2, 
\qquad 1 \le s \le l-1, \nonumber \\
\gamma_{s,s} (l)
\; = &\;
\sum_{j=1}^{l-1} \sum_{k=1}^N \langle j , \xi^k \rangle^2 \langle s , \xi^k \rangle^2, 
\qquad l \le s \le N, \nonumber \\
\gamma_{s,t} (l)
\; = &\; 
\sum_{j=1}^{l-1} \sum_{k=1}^N \langle j , \xi^k \rangle^2 \langle s , \xi^k \rangle \langle t , \xi^k \rangle ,
\qquad 1 \le s \ne t \le N, \nonumber \\
= & \; 
-\sum_{j=l}^{N} \sum_{k=1}^N \langle j , \xi^k \rangle^2 \langle s , \xi^k \rangle \langle t , \xi^k \rangle ,
\qquad 1 \le s \ne t \le N \label{coefficients de w l 0 p}
\end{align}
and
\begin{equation}\label{coefficients de u 1 1 p et u N N p}
\gamma_{s,t} (1) \; = \; 
\frac{1}{4} \sum_{k=1}^N \frac{\langle N, \xi^k \rangle^2}{\omega_k^2} \langle s,\xi^k \rangle \langle t, \xi^k \rangle
- 
\frac{1}{4} \sum_{k=1}^N \frac{\langle 1, \xi^k \rangle^2}{\omega_k^2} \langle s,\xi^k \rangle \langle t, \xi^k \rangle,
\qquad
1 \le s,t \le N.
\end{equation}

Formula \eqref{coefficients de u 1 1 p et u N N p} is directly derived from \eqref{expression for u l m z N en modes propres}, 
noting that $\gamma (1)$ is the only symmetric matrix such that $w_1 (0,p) = \sum_{s,t} \gamma_{s,t}(1) p_s p_t$.
To derive \eqref{coefficients de w l 0 p}, we observe that $\gamma(l)$ is the only symmetric matrix such that 
$w_l (0,p) = \sum_{s,t} \gamma_{s,t}(l) p_s p_t$. 
Starting from \eqref{expression for u l m z N en modes propres}, 
we deduce
\begin{equation*}
w_l (0,p) \; = \; \big( u_{l,l-1} - u_{l-1,l} \big) (0,p)
\; = \; 
\sum_{1 \le j \ne k \le N} 
\Big( 
\langle l ,\xi^j \rangle \langle l-1 , \xi^k \rangle - \langle l -1 ,\xi^j \rangle \langle l , \xi^k \rangle
\Big)
\frac{ \langle p , \xi^j \rangle \langle p ,\xi^k \rangle}{\omega_j^2 - \omega_k^2}.
\end{equation*}
For fixed BC, the eigenvectors $\xi^j$ satisfy the following relations for $1 \le j \le N$:
\begin{align*}
\langle \xi^j , 0 \rangle \; = \; \langle \xi^j , N+1 \rangle 
\; = &\; 0 \quad \text{(by definition)},\\
- \langle \xi^j , m-1 \rangle + (2 + \nu_m) \langle \xi^j , m \rangle - \langle \xi^j , m+1 \rangle 
\; = &\; \omega_j^2 \langle \xi^j , m \rangle, \quad 1 \le m \le N. 
\end{align*}
So the following recurrence relation is satisfied: 
\begin{equation}\label{recurrence relation for the eigenmodes}
\langle \xi^j , m+1 \rangle
\; = \; 
(2 + \nu_m - \omega_j^2) \langle \xi^j , m \rangle - \langle \xi^j , m-1 \rangle, \quad 1 \le m \le N.
\end{equation}

Let us first compute $w_2(0,p)$. Using \eqref{recurrence relation for the eigenmodes}, it comes
\begin{align*}
\langle 2 , \xi^j \rangle \langle 1 , \xi^k \rangle - \langle 1 , \xi^j \rangle \langle 2 , \xi^k \rangle
\; = &\; 
(2 + \nu_1 - \omega_j^2) \langle 1 , \xi^j \rangle \langle 1 , \xi^k \rangle
-
(2 + \nu_1 - \omega_k^2) \langle 1 , \xi^j \rangle \langle 1 , \xi^k \rangle \\
=&\;
- (\omega_j^2 - \omega_k^2) \langle 1 , \xi^j \rangle \langle 1 , \xi^k \rangle.
\end{align*}
Therefore 
\begin{align}
w_2 (0,p) \; = &\; - \sum_{1 \le j \ne k \le N} \langle 1 ,\xi^j \rangle \langle 1 , \xi^k \rangle \langle p , \xi^j \rangle \langle p , \xi^k \rangle 
\; = \; 
- \sum_{1 \le j , k \le N} \langle 1 , \xi^j \rangle \langle 1 , \xi^k \rangle \langle p , \xi^j \rangle \langle p , \xi^k \rangle
\; + \; \sum_{k=1}^N \langle 1 , \xi^k \rangle^2 \langle p , \xi^k \rangle^2 \nonumber\\
\; = &\; 
- \langle 1 , p \rangle^2 + \sum_{k=1}^N \langle 1 , \xi^k \rangle^2 \langle p , \xi^k \rangle^2, \label{expression for w 2}
\end{align}
where the last equality follows from the fact that $(\xi^k)_k$ forms an orthonormal basis. 

Let us now compute $w_l(0,p)$ for $2 < l \le N$. 
Again by \eqref{recurrence relation for the eigenmodes}, 
\begin{align*}
\langle l ,\xi^j \rangle \langle l-1 , \xi^k \rangle - \langle l-1 ,\xi^j \rangle \langle l , \xi^k \rangle
\; = &\; \phantom{-\;}
\Big( (2 + \nu_{l-1} - \omega_j^2) \langle l-1 ,\xi^j \rangle - \langle l-2 ,\xi^j \rangle \Big) \langle l-1 , \xi^k \rangle \\
&\; -\;
\langle l-1 ,\xi^j \rangle \Big( (2 + \nu_{l-1} - \omega_k^2) \langle l -1 , \xi^k \rangle - \langle l -2 , \xi^k \rangle \Big)\\
= &\; 
-\; (\omega_j^2 - \omega_k^2) \langle l-1 , \xi^j \rangle \langle l-1 , \xi^k \rangle \\
&\; + \; \langle l-1 , \xi^j \rangle \langle l -2 , \xi^k \rangle - \langle l-2 , \xi^j \rangle \langle l-1 , \xi^k \rangle.
\end{align*}
Therefore
\begin{align}
w_l(0,p) 
\; =& \;
\; - \sum_{1 \le j \ne k \le N} \langle l-1 , \xi^j \rangle \langle l-1 , \xi^k \rangle \langle p , \xi^j \rangle \langle p , \xi^k \rangle
\; + \; w_{l-1} (0,p) \nonumber \\
= & \; 
- \langle l-1 , p \rangle^2 \; +\; \sum_{k=1}^N \langle l-1 , \xi^k \rangle^2 \langle p , \xi^k \rangle^2 \; + \; w_{l-1} (0,p). \label{expression for w l, l bigger 2}
\end{align}

Combining \eqref{expression for w 2} and \eqref{expression for w l, l bigger 2}, 
we arrive to an expression valid for $2 \le l \le N$: 
\begin{equation*}
w_l (0,p) \; = \; 
\sum_{j=1}^{l-1} \left( \sum_{k=1}^N \langle j , \xi^k \rangle^2 \langle p , \xi^k \rangle^2 - \langle j , p \rangle^2 \right).
\end{equation*}
Let us now write $\langle j , p \rangle^2 = p_j^2$ and 
\begin{equation*}
\langle p , \xi^k \rangle^2
\; = \; \left( \sum_{s} p_s \langle s , \xi^k \rangle \right)^2 
\; = \; 
\sum_{s,t} p_s p_t \langle s , \xi^k \rangle \langle t , \xi^k \rangle.
\end{equation*}
We obtain
\begin{align*}
w_l (0,p) \; = &\; 
\sum_{s,t} p_s p_t \sum_{j=1}^{l-1} \sum_{k=1}^N \langle j , \xi^k \rangle^2 \langle s , \xi^k \rangle \langle t , \xi^k \rangle 
\; - \; \sum_{j=1}^{l-1} p_j^2 \\
= &\; \phantom{+\;}
\sum_{s=1}^{l-1}\, p_s^2 \left( 
\sum_{j=1}^{l-1} \sum_{k=1}^N \langle j , \xi^k \rangle^2 \langle s , \xi^k \rangle^2 - 1
\right) \\
&\; + \; 
\sum_{s=l}^N \, p_s^2 
\sum_{j=1}^{l-1} \sum_{k=1}^N \langle j , \xi^k \rangle^2 \langle s , \xi^k \rangle^2 
\\
& \; + \;
\sum_{1 \le s\ne t \le N} p_s p_t \sum_{j=1}^{l-1} \sum_{k=1}^N \langle j , \xi^k \rangle^2 \langle s , \xi^k \rangle \langle t , \xi^k \rangle .
\end{align*}

In this formula, the coefficients of $p_s^2$ coincide with $\gamma_{s,s}(l)$ given by \eqref{coefficients de w l 0 p} for $l \le s \le N$, 
and the coefficients of $p_s p_t$ with $s \ne t$ coincide with the first expression of $\gamma_{s,t}(l)$ given by \eqref{coefficients de w l 0 p}. 
To recover the coefficients $\gamma_{s,s}(l)$ for $1 \le s \le l-1$, just use the fact that $(\xi^k)_k$ and $(| k \rangle)_k$ are orthonormal basis: 
\begin{align*}
\sum_{j=1}^{l-1} \sum_{k=1}^N \langle j , \xi^k \rangle^2 \langle s , \xi^k \rangle^2 - 1
\; = &\; 
\sum_{k=1}^N \left( 1 - \sum_{j=l}^N \langle j , \xi^k \rangle^2 \right) \langle s , \xi^k \rangle^2 - 1 \\
\; = &\; 
\sum_{k=1}^N \langle s | \xi^k \rangle^2 - 1 - \sum_{j=l}^N \sum_{k=1}^N \langle j , \xi^k \rangle^2 \langle s , \xi^k \rangle^2 
\; = \; 
- \sum_{j=l}^N \sum_{k=1}^N \langle j , \xi^k \rangle^2 \langle s , \xi^k \rangle^2.
\end{align*}
The second expression for the coefficients $\gamma_{s,t}(l)$ with $s\ne t$ in \eqref{coefficients de w l 0 p} is obtained by a similar trick. 

\subsection{Exponential bounds}\label{subsec: exponential bounds}
We show here that there exist constants $\mathrm C < + \infty$ and $c > 0$ independent of $N$ such that 
\begin{equation}\label{exponential bound}
\Mean_\nu \left( \alpha_{j,k}^2 (l) \right) \; \le \; \mathrm C \, \exp \Big( - c \big( |j-l| + |k-l | \big) \Big), 
\quad 
\Mean_\nu \left( \gamma_{j,k}^2 (l) \right) \; \le \; \mathrm C \exp \Big( - c \big( |j-l| + |k-l | \big) \Big), 
\end{equation}
for $2 \le l \le N$ and for $1 \le j,k \le N$. 
This is still valid for $l=1$ if $|k-l|$ is replaced by $\min \{ |k-1|,|k-N|\}$ and $|j-l|$ by $\min \{ |j-1|,|j-N|\}$.
Due to \eqref{link between alpha and gamma}, it suffices to establish these bounds for the matrices $\gamma$.

Let us first observe that the almost sure bounds 
\begin{equation*}
|\gamma_{s,t}(l)| \; \le \; 1
\quad (2 \le l \le N),
\qquad
|\gamma_{s,t}(1)| \; \le \; \frac{1}{2 \min \{ \omega_k^2 : 1 \le k \le N\}}
\end{equation*}
hold for $1 \le s,t \le N$. 
This is directly deduced from \eqref{coefficients de w l 0 p} and \eqref{coefficients de u 1 1 p et u N N p} 
by taking absolute values inside the sums if needed, 
using that $(\xi^k)_k$ and $(| k \rangle)_k$ are orthonormal basis, and Cauchy-Schwarz inequality if needed.
By \eqref{bounds eigenvalues}, $\min \{ \omega_k^2 : 1 \le k \le N\} \ge c > 0$, where $c$ does not depend on $N$.
In particular $\Mean_\nu \left( |\gamma_{s,t}|^p \right) \le \mathrm C_p \Mean_\nu |\gamma_{s,t}|$ for every $p \ge 1$, so that 
we only need to bound $ \Mean_\nu |\gamma_{s,t}|$.

We now will apply localization results originally derived by Kunz and Souillard (\cite{kun}), but we follow the exposition {\comm{given}} by \cite{dam}.
From \eqref{coefficients de w l 0 p} and \eqref{coefficients de u 1 1 p et u N N p}, we see that we are looking for upper bound on the absolute value of sums of the type
\begin{equation*}
\sum_k \langle r ,\xi^k \rangle^2 \langle t ,\xi^k \rangle^2, \quad r < t,
\end{equation*}
and of the type 
\begin{equation*}
\sum_k \langle r , \xi^k \rangle^2 \langle s , \xi^k \rangle \langle t , \xi^k \rangle, 
\quad
\sum_k \langle r , \xi^k \rangle \langle s , \xi^k \rangle^2 \langle t , \xi^k \rangle, 
\quad
\sum_k \langle r , \xi^k \rangle \langle s , \xi^k \rangle \langle t , \xi^k \rangle^2, 
\qquad 
r < s < t.
\end{equation*}
Since $|\langle r, \xi^k\rangle| \le 1$ for $1 \le r,k \le N$, all of them can be bounded by 
\begin{equation*}
\sum_k | \langle r , \xi^k \rangle \langle t , \xi^k \rangle |.
\end{equation*}
By the formula before Lemma 4.3 in \cite{dam}, and the lines after the proof of this lemma, 
we may conclude that there exist constants $\mathrm C < + \infty$ and $c > 0$ independent of $N$ such that
\begin{equation*}
\Mean_\nu \Big( \sum_k | \langle r , \xi^k \rangle \langle t , \xi^k \rangle | \Big) \; \le \; \mathrm C \, \ed^{-c(t-r)}. 
\end{equation*}
Together with the remarks formulated up to here, this allows to deduce \eqref{exponential bound}.

\subsection{Concluding the proof of Lemma \ref{lemm: Poisson equation}}\label{subsec: some integrals}
We write
\begin{equation*}
\Mean_\nu \mu_T ( u_N^2 ) 
\; = \; 
\frac{1}{4 N} \sum_{m,n} \Mean_\nu \mu_T (w_{m} \multiplication w_{n})
\quad \text{and} \quad
\Mean_\nu \mu_T \big( (A_{anh}u_N)^2 \big) 
\; = \; 
\frac{1}{4 N} \sum_{m,n} \Mean_\nu \mu_T ( A_{anh} w_{m} \multiplication A_{anh} w_{n}).
\end{equation*}
We will establish that there exist constants $\mathrm C < +\infty$ and $c >0$ such that 
\begin{equation}\label{decroissance exponentielle des w l}
| \Mean_\nu \mu_T (w_{m} \multiplication w_{n}) | \; \le \; \mathrm C \, \ed^{-c |m-n|}
\quad \text{and} \quad
| \Mean_\nu \mu_T ( A_{anh} w_{m} \multiplication A_{anh} w_{n}) | \; \le \; \mathrm C \, \ed^{-c |m-n|}
\end{equation}
for $1 \le m,n\le N$. This will conclude the proof. 

Let us fix $1 \le m,n \le N$. 
Let us first consider $| \Mean_\nu \mu_T (w_{m} \multiplication w_{n}) |$.
Let us observe that 
the functions $w_l$ are of zero mean by construction, and so the relation 
\begin{equation}\label{zero mean condition}
\sum_{j} \gamma_{j,j}(l) \int p_j^2 \, \dd \mu_T + \sum_{j,k} \alpha_{j,k} (l) \int q_j q_k \, \dd \mu_T + c(l) \; = \; 0
\end{equation}
holds for $1 \le l \le N$.
Using this relation, it is computed that
\begin{align*}
\mu_T (w_{m} \multiplication w_{n})
\; = &\;
\mu_T\bigg(
\Big(
\sum_{i,j} \alpha_{i,j}(m) q_i q_j + \sum_{i,j} \gamma_{i,j}(m) p_i p_j + c(m)
\Big) 
\Big(
\sum_{i,j} \alpha_{i,j}(n) q_i q_j + \sum_{i,j} \gamma_{i,j}(n) p_i p_j + c(n)
\Big) 
\bigg)\\
=&\;
\sum_{i,j,k,l} \alpha_{i,j}(m) \alpha_{k,l}(n) \int q_i q_j q_k q_l \, \dd \mu_T
+ 
\sum_{i,j,k,l} \alpha_{i,j}(m) \gamma_{k,l}(n) \int q_i q_j p_k p_l \, \dd \mu_T \\
&+
\sum_{i,j,k,l} \gamma_{i,j}(m) \alpha_{k,l}(n) \int p_i p_j q_k q_l \, \dd \mu_T
+
\sum_{i,j,k,l} \gamma_{i,j}(m) \gamma_{k,l}(n) \int p_i p_j p_k p_l \, \dd \mu_T
- 
c(m) c(n) \\
=&\; 
S_1 + S_2 + S_3 + S_4 - c(m) c(n).
\end{align*}

Using \eqref{zero mean condition} and the fact that $\int p^2 \, \dd \mu_T = T$, the sum $S_1$ is rewritten as 
\begin{align*}
S_1 
\; = &\;
\sum_{i,j,k,l} \alpha_{i,j}(m) \alpha_{k,l}(n) \int \Big( q_i q_j - \int q_i q_j \, \dd \mu_T\Big) \Big( q_k q_l - \int q_k q_l \, \dd \mu_T\Big) \, \dd \mu_T \\
& + 
\sum_{i,j,k,l} \alpha_{i,j}(m) \alpha_{k,l}(n) \int q_i q_j \, \dd \mu_T \int q_k q_l \, \dd \mu_T \\
= &\;
\sum_{i,j,k,l} \alpha_{i,j}(m) \alpha_{k,l}(n) \int \Big( q_i q_j - \int q_i q_j \, \dd \mu_T\Big) \Big( q_k q_l - \int q_k q_l \, \dd \mu_T\Big) \, \dd \mu_T \\
& +
T^2 \sum_{i,j} \gamma_{i,i}(m) \gamma_{j,j}(n) + T c(m) \sum_i \gamma_{i,i}(n) + T c(n) \sum_i \gamma_{i,i}(m) + c(m)c(n).
\end{align*}
Then, still using \eqref{zero mean condition}, we get
\begin{equation*}
S_2 + S_3 \; = \; 
- T c(m) \sum_i \gamma_{i,i}(n) 
- T c(n) \sum_i \gamma_{i,i}(m)
- 2 T^2 \sum_{i,j} \gamma_{i,i}(m) \gamma_{j,j}(n). 
\end{equation*}
Finally, the terms in the sum $S_4$ are non zero only when 
\begin{equation*}
i=j=k=l, 
\quad
i=j,k=l,i\ne k,
\quad
i=k,j=l,i\ne j,
\quad
i=l,j=k,i\ne j.
\end{equation*}
Using that $\int p^4 \dd \mu_T = 3 (\int p^2 \, \dd \mu_T)^2$ and that $\int p^2 \, \dd \mu_T = T$, 
$S_4$ is seen to be equal to 
\begin{equation*}
S_4 
\; = \; 
T^2 \sum_{i,j} \Big( 
\gamma_{i,i}(m) \gamma_{j,j}(n) + 2 \gamma_{i,j}(m) \gamma_{i,j}(n)
\Big).
\end{equation*}

Therefore
\begin{align*}
\mu_T (w_{m} \multiplication w_{n})
\; = &\;
\sum_{i,j,k,l} \alpha_{i,j}(m) \alpha_{k,l}(n) \int \Big( q_i q_j - \int q_i q_j \, \dd \mu_T\Big) \Big( q_k q_l - \int q_k q_l \, \dd \mu_T\Big) \, \dd \mu_T \\ 
&+ 2T^2 \sum_{i,j} \gamma_{i,j}(m) \gamma_{i,j}(n).
\end{align*}
Applying the decorrelation bound \eqref{decay of correlations Bodineau Helffer} and the exponential estimate \eqref{exponential bound}, 
the result is obtained.

Let us next consider $| \Mean_\nu \mu_T ( A_{anh} w_{m} \multiplication A_{anh} w_{n}) |$.
We find from \eqref{def A anh u 1} that
\begin{equation*}
\mu_T ( A_{anh} w_{m} \multiplication A_{anh} w_{n})
\; = \; 
\sum_{s,t} \mu_T \big( \phi_t(q,m) \, \phi_s (q,n) \, p_s p_t \big)
\; = \; 
T \sum_t \mu_T \big( \phi_t(q,m) \, \phi_t (q,n) \big).
\end{equation*}
Now, it follows from \eqref{def A anh u 2} that
$\phi_t (q,k) \; = \; \sum_s \gamma_{s,t} (k) \rho_s (q)$
for $1 \le k \le N$.
Here $\rho_s (q) = \rho_s (q_{s-1},q_s,q_{s+1})$ is a function of mean zero since the potentials $U$ and $V$ are symmetric. 
We write
\begin{equation*}
\mu_T ( A_{anh} w_{m} \multiplication A_{anh} w_{n})
\; = \; 
\sum_t \sum_{s,s'}
\gamma_{s,t} (m) \, \gamma_{s',t} (n) \, \mu_T (\rho_s \multiplication \rho_{s'}).
\end{equation*}
Applying the decorrelation bound \eqref{decay of correlations Bodineau Helffer} and the exponential estimate \eqref{exponential bound} yield the result. 
$\square$

\section{Convergence results}\label{sec: convergence}
In this section we show the convergence result \eqref{convergence of Green Kubo}.
We assume thus $\lambda >0$ and $\lambda' = 0$.

We start with some definitions (see \cite{ber} for details).
The dynamics defined in Section \ref{sec: model and results} can also be defined for a set of particles indexed in $\Z$ instead of $\Z_N$. 
Points on the phase space are written $x = (q,p)$, with $q = (q_k)_{k\in \Z}$ and $p = (p_k)_{k \in \Z}$. 
Let us denote by $\mathcal L$ the generator of this infinite-dimensional dynamics.
We remember here that $\mu_T^{(N)}$ represents the Gibbs measure of a system of size $N$ ; we denote by $\mu_T^{(\infty)}$ the Gibbs measure of the infinite system
(the dependence on the size will still be dropped in the cases where it is irrelevant). 
We extend the definition \eqref{local current} of local currents $j_{k}$ to all $k\in \Z$ ($j_k = j_{k,har}$ since $\lambda' = 0$).
If $u = u (x,\nu)$, with $\nu = (\nu_k)_{k \in \Z}$ a sequence of pinnings, and if $k \in \Z$, we write $\tau_k u (x,\nu) = u(\tau_k x , \tau_k \nu)$, where
\begin{equation*}
(\tau_k q)_j \; = \; q_{k+j}, 
\quad
(\tau_k p)_j \; = \; p_{k+j},
\quad
(\tau_k \nu)_j \; = \; \nu_{k+j}. 
\end{equation*}
Finally, we denote by $\ll \cdot, \cdot \gg$ the inner-product defined, for local bounded functions $u$ and $v$, by
\begin{equation*}
\ll u, v \gg \; = \; \sum_{k \in \ZZ} \Mean_\nu \big(  \mu_T^{(\infty)} (u \multiplication \tau_k v^*) -\mu_T^{(\infty)} (u) \mu_T^{(\infty)} (v) \big)
\end{equation*}
where $v^*$ is the complex conjugate of $v$,
and by ${\mc H}$ the corresponding Hilbert space, obtained by completion of the bounded local functions. 

We start with two lemmas.
We have no reason to think that Lemma \ref{lemm: powers of L are bounded} still holds if an anharmonic potential is added, 
and this is the main reason why we here restrict ourselves to harmonic interactions. 
\begin{Lemma}\label{lemm: powers of L are bounded}
There exists a constant $\mathrm C < +\infty$ such that, 
for any realization of the pinnings, 
for the finite dimensional dynamics with free or fixed B.C., 
or for the infinite dynamics, 
for any $k\ge 1$ and for any $l \in \Z_N$ (resp. $k \in \Z$ for the infinite dynamics),
\begin{equation*}
\| L^k j_l \|_{\Lp^2 (\mu_T)} \; \le \; \mathrm C^k 
\qquad 
\text{(resp. }\| \mathcal L^k j_l \|_{\Lp^2 (\mu_T)} \; \le \; \mathrm C^k \text{).}
\end{equation*}
\end{Lemma}
\Proof 
Let us consider the infinite dimensional dynamics ; other cases are similar. 
We can take $l=0$ without loss of generality. 
The function $j_0$ is of the form $j_0 = \langle q, \alpha q \rangle + \langle q, \beta p \rangle + \langle p , \gamma p \rangle$, 
with $\alpha = \gamma = 0$ and $\beta$ defined by 
\begin{equation*}
\beta_{i,j} \; = \; \frac{1}{2} \big( \delta_{0,0}(i,j) - \delta_{(1,1)}(i,j) + \delta_{0,1}(i,j) - \delta_{1,0}(i,j) \big).
\end{equation*}
Now, if $u$ is any function of the type $u = \langle q, \alpha q \rangle + \langle q, \beta p \rangle + \langle p , \gamma p \rangle$, 
then $\mathcal L u = \langle q, \alpha' q \rangle + \langle q, \beta' p \rangle + \langle p , \gamma' p \rangle$ with
\begin{equation*}
(\alpha',\beta',\gamma') 
\; = \;
\Big(
- \frac{\beta\Phi + \Phi\beta^\dagger}{2} , \alpha - 2 \Phi \gamma - 2 \lambda \beta, \frac{\beta + \beta^\dagger}{2} - 4 \lambda \tilde \gamma
\Big),
\end{equation*}
where $\tilde \gamma$ is such that $(\tilde \gamma)_{i,i} = 0$ and $(\tilde \gamma )_{i,j} = \gamma_{i,j}$ for $i \ne j$. 
Thus 
\begin{equation*}
\mathcal L^k j_0 \; = \; \langle q, \alpha_{(k)} q \rangle + \langle q, \beta_{(k)} p \rangle + \langle p , \gamma_{(k)} p \rangle, 
\end{equation*}
and there exists a constant $\mathrm C < + \infty$ such that $\zeta_{i,j} = 0$ whenever $|i| \ge \mathrm C k$ or $|j| \ge \mathrm C k$ 
and such that $|\zeta_{i,j}| \le \mathrm C^k$ otherwise, with $\zeta$ one of the three matrices $\alpha_{(k)}$, $\beta_{(k)}$ or $\gamma_{(k)}$. 
The claim is obtained by expressing $\| \mathcal L^k j_l \|_{\Lp^2 (\mu_T)} $ in terms of the matrices $\alpha_{(k)}$, $\beta_{(k)}$ and $\gamma_{(k)}$.
$\square$

\noindent
Explicit representation for the matrix $\Phi^{-1}$ in Lemma 1.1. in \cite{bry} allows to deduce the following lemma. 
\begin{Lemma}\label{lemm: approximation measure}
Let $f$ and $g$ be two polynomials of the type $\langle q, \alpha q \rangle + \langle p, \beta p \rangle + \langle p , \gamma p \rangle$, 
and assume that there exists $n\in \N$ such that $\alpha_{i,j} = \beta_{i,j} = \gamma_{i,j} = 0$ whenever $|i|> n$ or $|j|> n$. 
Then there exists $c >0$ such that, for fixed or periodic B.C.,
\begin{equation*}
\big|\mu_T^{(N)} ( f \multiplication  g^* ) - \mu_T^{(\infty)}( f \multiplication g^* )\big| \; = \; \mathcal O ( \ed^{-cN}) \quad \text{as} \quad N \rightarrow \infty.
\end{equation*}
\end{Lemma}

Let then $\mathbb D = \{ z \in \CC : \Re z > 0 \}$. 
For every $z\in \mathbb D$, let $u_z$ be the unique solution to the {\comm{resolvent equation in ${\mc H}$}}
\begin{equation}\label{def of the function u z}
(z - \mathcal L) u_z \; = \; j_{0}.
\end{equation}
We know from Theorem 1 in \cite{ber}\footnote{
The model studied there is not exactly the same.
The proof of the properties we mention here can be however readily adapted.
}, and from its proof, that 
\begin{equation}
\label{eq:tibere}
\lim_{z \to 0} \ll u_z, j_{0} \gg \quad \text{exists and is finite}
\end{equation}
and that 
\begin{equation}
\label{eq:trajan}
\lim_{z\rightarrow 0} z \ll u_z, u_z \gg \; = \; 0.
\end{equation}
%
For $z \in \mathbb D$ and $N \ge 3$, let $u_{k,z,N}$ be the unique solution to the equation
\begin{equation}\label{def of the function u k z N}
(z- L )u_{k,z,N} \; = \; j_k.
\end{equation}
so that 
\begin{equation}\label{def of the function u z N}
u_{z,N} \; := \; \frac{1}{\sqrt N} \sum_k u_{k,z,N} \qquad \text{solves} \qquad (z - L) u_{z,N} = \mathcal J_N.
\end{equation}
\begin{Lemma}\label{lemm: almost sure convergence}
For fixed or free boundary conditions and for almost all realizations of the pinnings, 
\begin{align*}
\lim_{z\rightarrow 0}\lim_{N\rightarrow \infty} \mu_T \big( u_{z,N} \multiplication \mathcal J_N \big)
\; &= \;
\lim_{z\rightarrow 0}\lim_{N\rightarrow \infty} \Mean_\nu\mu_T \big( u_{z,N} \multiplication \mathcal J_N \big)
\; = \;
\lim_{z \to 0} \ll u_z, j_{0} \gg , \\
\lim_{z\rightarrow 0}\lim_{N\rightarrow \infty} z\, \mu_T \big( u_{z,N} \multiplication u_{z,N} \big)
\; &= \; 
\lim_{z\rightarrow 0}\lim_{N\rightarrow \infty} z\, \Mean_\nu\mu_T \big( u_{z,N} \multiplication u_{z,N} \big)
\; = \; 0 .
\end{align*}
\end{Lemma}
\Proof 
By \eqref{eq:tibere} and \eqref{eq:trajan},
it suffices to establish separately that, for every $z \in \mathbb D$, and for almost every realization of the pinnings,
\begin{align*}
&\lim_{N\rightarrow \infty} \mu_T \big( u_{z,N} \multiplication \mathcal J_N \big) \; = \; \ll u_z, j_{0} \gg ,
\qquad
\lim_{N\rightarrow \infty} \Mean_\nu\mu_T \big( u_{z,N} \multiplication \mathcal J_N \big) \; = \; \ll u_z, j_{0} \gg ,\\
&\lim_{N\rightarrow \infty} z\, \mu_T \big( u_{z,N} \multiplication u_{z,N} \big) \; = \; 
z \ll u_z,u_z \gg, 
\qquad
\lim_{N\rightarrow \infty} z\, \Mean_\nu\mu_T \big( u_{z,N} \multiplication u_{z,N} \big) \; = \; z \ll u_z,u_z \gg.
\end{align*}
The proof of these four relations is in fact very similar, and we will focus on the first one. 
We proceed in two steps: we first show the result for $|z|$ large enough, and then extend it to all $z\in \mathbb D$. 

\textbf{First step.}
Here we fix $z \in \mathbb D$ with $|z|$ large enough. 
We first assume periodic boundary conditions.
The function $u_z$ solving \eqref{def of the function u z} may be given by 
\begin{equation*}
u_z \; = \; \sum_{k \ge 0} z^{-(k+1)} \mathcal L^k j_0,
\end{equation*}
this series converging in virtu of Lemma \ref{lemm: powers of L are bounded} for $|z|$ large enough. 
Let now $n \ge 1$.
We compute
\begin{equation}\label{une equation}
\ll j_0, u_z \gg \; = \; \sum_{k = 0}^n z^{-(k+1)} \sum_{l \in \Z} \Mean_\nu\mu_T^{(\infty)} ( j_l \multiplication L^k j_0 )
\, + \, 
\sum_{k = n+1}^\infty z^{-(k+1)} \sum_{l \in \Z} \Mean_\nu \mu_T^{(\infty)} ( j_l \multiplication L^k j_0 ).
\end{equation}
For every given $k$, the sum over $l$ is actually a sum over $\mathrm C\, k$ non-zero terms only, for some $\mathrm C < + \infty$. 
From this fact and from Lemma \ref{lemm: powers of L are bounded}, 
it is concluded that the second sum in the right hand side of \eqref{une equation} converges to 0 as $n \rightarrow \infty$. 
Similarly we write
\begin{equation}\label{une autre equation}
\mu_T^{(N)} ( \mathcal J_N \multiplication u_{z,N} ) \; = \;
\frac{1}{N} \sum_{s,t} \sum_{k=0}^n z^{-(k+1)} \mu_T^{(N)} ( j_s \multiplication L^k j_t )
\, + \, 
\frac{1}{N} \sum_{s,t} \sum_{k=0}^n z^{-(k+1)} \mu_T^{(N)} ( j_s \multiplication L^k j_t ).
\end{equation} 
Here as well, the second term in \eqref{une autre equation} is such that 
\begin{equation*}
\lim_{n\rightarrow \infty} \limsup_{N \rightarrow \infty} \frac{1}{N} \sum_{s,t} \sum_{k=0}^n z^{-(k+1)} \mu_T^{(N)} ( j_s \multiplication L^k j_t ) \; = \; 0.
\end{equation*}
To handle the first term in \eqref{une autre equation}, let us write 
\begin{equation*}
F_n (\nu) \; = \; \sum_{t \in \Z_N} \sum_{k=0}^n z^{-(k+1)} \mu_T^{(N)} ( j_0 \multiplication L^k j_t ). 
\end{equation*}
Then in fact
\begin{equation*}
\frac{1}{N} \sum_{s,t} \sum_{k=0}^n z^{-(k+1)} \mu_T^{(N)} ( j_s \multiplication L^k j_t )
\; = \; 
\frac{1}{N} \sum_{s} F_n (\tau_s \nu).
\end{equation*}
The result is obtained by letting $N \rightarrow \infty$, invoking Lemma \ref{lemm: approximation measure} and the ergodic theorem, 
and then letting $n \rightarrow \infty$. 
If we had started with fixed boundary conditions, then, for every fixed $n$, 
all the previous formulas remain valid up to some border terms that vanish in the limit $N \rightarrow \infty$ due to the factor $1/N$.

\textbf{Second step.}
Denote by ${\mf L}_{N,\nu} (z)$ and ${\mf L} (z)$ the complex functions defined on ${\bb D}$ by
\begin{equation*}
{\mf L}_{N,\nu} (z) \; = \; \mu_T^{(N)} ( u_{z,N} \multiplication \mathcal J_N )
\quad \text{and} \quad 
{\mf L} (z) \; = \; \ll u_z, j_{0}\gg.
\end{equation*} 
The first observation is that these functions are well defined and analytic on ${\bb D}$. 
Moreover, similarly to what is proved in \cite{ber}, they are uniformly bounded on ${\bb D}$ by a constant independent of $N$ and the realization of the pinning $\nu$.

Let us fix a realization of the pinnings. The family $\{ \mf L_{N,\nu} \; ;\; N \ge 1\}$ is a normal family and 
by Montel's Theorem we can extract a subsequence $\{\mf L_{N_{k}, \nu}\}_{k \ge 1}$ 
such that it converges (uniformly on every compact set of ${\bb D}$) to an analytic function $f^{\star}_{\nu}$. 

By the first step we know that $f_{\nu}^{\star} (z) ={\mf L} (z)$ for any real $z>z_0$. 
Thus, since the functions involved are analytic, $f^{\star}_{\nu}$ coincides with ${\mf L} $ on ${\bb D}$. 
It follows that the sequence $\{{\mf L}_{N,\nu} (z)\}_{N \ge 1}$ converges for any $z \in \bb D$ to ${\mf L} (z)$.
$\square$

Following a classical argument, we can now proceed to the

\ProofOf{of \eqref{convergence of Green Kubo}}
For any $z > 0$, it holds that
\begin{align*}
\frac{1}{\sqrt{t}} \int_0^t \mathcal J_N \circ X^s \, \dd s
\; = & \; 
-\frac{1}{\sqrt t} \int_0^t L u_{z,N} \circ X^s \, \dd s + 
\frac{z}{\sqrt t} \int_0^t u_{z,N} \circ X^s \, \dd s \\
\; = & \;
\frac{1}{\sqrt t} \mathcal M_{z,N,t} 
- 
\frac{u_{z,N}\circ X^t - u_{z,N}}{\sqrt t} 
+ 
\frac{z}{\sqrt t} \int_0^t u_{z,N} \circ X^s \, \dd s
\end{align*}
Here $\mathcal M_{z,N,t}$ is a stationary martingale {\comm{with variance given by 
\begin{equation}\label{formula for the convergence of GK abc}
\mu_T\Mean \big( \mathcal M_{z,N,t}^2 \big)
\; = \; 
\mu_T \big( u_{z,N} \multiplication (z - L) u_{z,N} \big) - z\, \mu_T (u_{z,N} \multiplication u_{z,N}).
\end{equation}
Here the equality $\mu_T \big( u_{z,N} \multiplication A_{har} u_{z,N} \big) = 0$ has been used.}}
Next
\begin{equation*}
\mu_T\Mean \left(
\frac{u_{z,N} \circ X^{t} - u_{z,N}}{\sqrt t}
\right)^2
\; \le \; 
\frac{2}{t} \, \mu_T (u_{z,N} \multiplication u_{z,N})
\end{equation*}
and
\begin{equation*}
\mu_T\Mean \left(
\frac{z}{\sqrt t} \int_0^t u_{z,N} \circ X^s \, \dd s
\right)^2\; \le \;
z^2 t \, \mu_T (u_{z,N} \multiplication u_{z,N}).
\end{equation*}
Reminding that $\mu_T \big( u_{z,N} \multiplication (z - L) u_{z,N} \big) = \mu_T \big( u_{z,N} \multiplication \mathcal J_N \big)$, 
the proof is completed by taking $z = 1/t$ and invoking Lemma \ref{lemm: almost sure convergence}. $\square$
\section{Lower bound in the absence of anharmonicity}\label{sec: lower bound}
We here establish the lower bound in \eqref{upper and lower bound on Green Kubo}, and so we assume $\lambda >0$ and $\lambda' = 0$. 
We also assume periodic boundary conditions.
We use the same method as in \cite{ber} (see also \cite{dha}). 
According to Section \ref{sec: convergence}, it is enough to establish that there exists a constant $c>0$ such that, 
for almost every realization of the pinnings, for every $z >0$ and for every $N \ge 3$, 
\begin{equation}\label{equivalent to lower bound}
\mu_T \big( \mathcal J_N \multiplication ( z - L)^{-1} \mathcal J_N \big) \; \ge \; c.
\end{equation} 
Indeed, by \eqref{def of the function u z N}, 
$\mu_T \big( \mathcal J_N \multiplication ( z - L)^{-1} \mathcal J_N \big) = \mu_T \big( u_{z,N} \multiplication ( z - L) u_{z,N} \big)$, 
and, by \eqref{formula for the convergence of GK abc}, this quantity converges to the right hand side of \eqref{convergence of Green Kubo}.

\ProofOf{of \eqref{equivalent to lower bound}}
For periodic B.C., the total current $J_N$ is given by
\begin{equation*}
J_N = \cfrac{1}{2} \sum_{k \in \Z_N} (q_{k} p_{k+1} - q_{k+1} p_k).
\end{equation*} 
To get a lower bound on the conductivity, we use the following variational formula
\begin{equation}
\label{eq:vf0}
\mu_T \left( J_N \, (z-L)^{-1} \, J_N \right) \; = \; \sup_{f} \left\{ 2 \mu_T (J_N \multiplication f) - 
\mu_T \left( f \multiplication (z-\lambda S) f\right) - \mu_T \left( A_{har} f \multiplication (z-\lambda S)^{-1} \, A_{har} f \right) \right\}
\end{equation} 
where the supremum is carried over the test functions $f\in \mathcal C^{\infty}_{temp}(\R^{2N})$.
See \cite{set} for a proof. 
We take $f$ in the form
\begin{equation*}
f \; = \; a \langle q, \beta p \rangle \quad \text{with} \quad a \in \R \quad \text{and} \quad \beta \;=\; \Phi M,
\end{equation*} 
where $M$ is the antisymmetric matrix such that $M_{i,j}= \delta_{i,j-1} - \delta_{i,j+1}$,
with the convention of periodic B.C.: $\delta_{1,N+1} = \delta_{1,1}$ and $\delta_{N,0} = \delta_{N,N}$. 

First, we have
\begin{equation*}
A_{har} f 
\; = \; 
a\langle p, \beta p \rangle - a\langle q, \beta \Phi q \rangle 
\; = \; 
a \sum_{i\ne j} \beta_{i,j} p_{i} p_j
\end{equation*} 
since $\beta \Phi = \Phi M \Phi$ is antisymmetric, 
and since $\beta_{i,i} = 0$ for $1\le i \le N$. 
Since $S(p_i p_j)= -4 p_i p_j$ for $i \ne j$, we obtain
\begin{align}
\mu_T \left( Af \multiplication (z-\lambda S)^{-1} \, Af \right) 
\; = &\; 
\frac{1}{z + 4 \lambda} \mu_T (A_{har}f \multiplication A_{har}f)
\; = \; 
\frac{a^2 T^2}{z + 4\lambda} \sum_{i\ne j} \big( \beta_{i,j}^2 + \beta_{i,j}\beta_{j,i} \big)
\nonumber\\
\; \le &\; 
\mathrm C \,\frac{ a^2 T^2 N}{z + 4\lambda}
\label{eq:Aterm}
\end{align}
for some constant $\mathrm C < + \infty$.
Next, 
since $S p_k = -2 p_k$ for $ \le k \le N$, there exists some constant $\mathrm C < + \infty$ such that
\begin{align}
\mu_T ( f \multiplication (z-\lambda S) f ) 
\; =&\; a^2 T^2 (z+2\lambda) \sum_{i,j} \beta_{i,j} (\Phi^{-1} \beta)_{i,j} = a^2 T^2 (z+2\lambda) {\rm{Tr}} \left[ \beta^{\dagger} \Phi^{-1} \beta \right]
\nonumber \\
\; \le &\; \mathrm C a^2 T^2 (z+2\lambda) N.
\label{eq:Sterm}
\end{align}
Let us finally estimate the term $\mu_T (J_N \, f)$: 
\begin{align}
\mu_T (J_N \, f) \; = &\;
\frac{a}{2} \mu_T\Big(
\sum_{i,j} \beta_{i,j} q_i p_j \multiplication \sum_{k\in \Z_N} (q_k p_{k+1} - q_{k+1}p_k)
\Big)
\; = \;
\frac{aT^2}{2} \sum_{i,k} \big( \beta_{i,k+1} \mu_T ( q_i q_k ) - \beta_{i,k} \mu_T ( q_i q_{k+1} ) \big) 
\nonumber\\
= &\;
\frac{a T^2}{2} \sum_{k \in \Z_N} \big( (\beta^\dagger \Phi^{-1})_{k+1,k} - (\beta^\dagger \Phi^{-1})_{k+1,k} \big)
\; = \;
\frac{a T^2}{2} \sum_{k \in \Z_N} \big( M_{k,k+1} - M_{k+1,k} \big) 
\nonumber\\
= &\;
a T^2 N.
\label{eq:Jterm}
\end{align}

By \eqref{eq:Aterm}, \eqref{eq:Sterm}, \eqref{eq:Jterm} and the variational formula \eqref{eq:vf0}, 
we find that there exists a constant $\mathrm C < + \infty$, 
independent of the realization of the disorder, of $\lambda$ and of $N$, such that for any positive $a$,
\begin{equation*}
\frac{1}{N T^2} \mu_T (J_N \, (z- L)^{-1} \, J_N) \; \ge \; a - \mathrm C {a^2}\left( (z+2\lambda)+ \cfrac{1}{z+ 4 \lambda} \right).
\end{equation*}
By optimizing over $a$, this implies
\begin{equation*}
\frac{1}{N T^2} \mu_T (J_N \, (z- L)^{-1} \, J_N) \; \ge \; \frac{1}{4 \mathrm C} \left( (z+2\lambda)+ \cfrac{1}{z+ 4 \lambda}\right)^{-1} .
\end{equation*}
Since $\mathcal J_N = J_N/\sqrt N$, this shows (6.1). $\square$

\textbf{Acknowledgements.}
We thank J.-L. Lebowitz and J. Lukkarinen for their interest in this work and T. Bodineau for useful discussions. 
We thank C. Liverani, S. Olla and L.-S. Young, 
as organizers of the Workshop on the Fourier Law and Related Topics, 
as well as the kind hospitality of the Fields Institute in Toronto, 
where this work was initiated.
C.B. acknowledges the support of the French Ministry of 
Education through the grants ANR-10-BLAN 0108.
F.H. acknowledges the European Advanced Grant Macroscopic Laws and Dynamical Systems (MALADY) (ERC AdG 246953) for financial support.

\end{document}